\documentclass[a4paper]{article}

\usepackage{fullpage}
\usepackage{palatino}
\usepackage[utf8x]{inputenc}
\usepackage[T1]{fontenc}
\usepackage{wangtile}
\usepackage{hyperref}
\usepackage{amsmath, amsfonts, amssymb}

\newcommand{\sk}{\medskip}

\newcommand{\floor}[1]{\lfloor #1 \rfloor}
\newcommand{\ceil}[1]{\lceil #1 \rceil}

\usepackage{amsthm}

\newcounter{cnt}

\newtheorem{theorem}[cnt]{Theorem}
\newtheorem{lemma}[cnt]{Lemma}
\newtheorem{prop}[cnt]{Proposition}

\theoremstyle{definition}

\theoremstyle{plain}

\usepackage{tikz, color, calc, caption, subcaption, wangtile}
\usetikzlibrary{calc}

\tikzstyle{vertex}=[draw=black,circle,minimum size=1cm]
\tikzstyle{write0}=[-latex,thick,color=red]
\tikzstyle{write1}=[-latex,thick,color=blue]
\tikzstyle{write2}=[-latex,thick,color=green]

\definecolor{turquoise}{rgb}{0.25,0.87,0.81}

\newcounter{a}
\newcounter{b}

\usepackage[colorinlistoftodos]{todonotes}

\begin{document}

\title{Aperiodic tilings and entropy\footnote{supported by ANR project EMC NT09 555297}}
\author{Bruno Durand\thanks{Universit\'e Montpellier 2, Lirmm 161 rue Ada, 34095 Montpellier cedex 5 France, \url{www.lirmm.fr/\~bdurand}} , Guilhem Gamard\thanks{Lirmm and ENS Paris}, Anael Grandjean\thanks{Lirmm and ENS Lyon}}
\date{}
\maketitle

\begin{abstract}
In this paper we present a construction of Kari-Culik aperiodic tile set — the smallest known until now. With the help of this construction, we prove that this tileset has positive entropy. We also explain why this result was not expected.
\end{abstract}

\section{Introduction}

In this paper we focus on aperiodic tilesets. These tilesets can be used to tile the plane but none of the obtained tilings has a period.  The role of aperiodic tilesets is crucial in different fields such as logics (see for instance~\cite{gurevich}) or for the study of quasi-periodic structures such as quasi-crystals. Furthermore these aperiodic tilesets are a classical tool to prove undecidability problems for planar structures or dynamical systems. We work with the formalization that Wang proposed in~\cite{Wang}.

A classical question about a tileset is to measure its entropy. Roughly speaking the entropy of a tileset is positive if "points of freedom are dense". One can easily make it positive for any aperiodic tileset by a cartesian product with a free bit. The number of tiles is then multiplicated by a factor two but the resulting tiling has positive entropy. It is easy to observe that for classical self similar tilesets such as Berger~\cite{BergerPhD}, Robinson~\cite{Robinson} the entropy is zero. The main question we adress is the entropy of the smallest known aperiodic tileset : it was conjectured that its entropy is zero but we prove it is positive. This entropy zero conjecture comes from other works on this tileset and some algorithmic remarks developped in section~\ref{sec:choices}.

Our paper is organized as follows: first we explain exactly the same tileset as Kari and Culik in ~\cite{Kari,Culik}. Our explanation makes it easier to analyse (repeating the proof of aperiodicity). Then in section~\ref{sec:entropy} we formulate a substitutive property that guarantee positive entropy of Kari-Culik tileset. The rest of the section is devoted to the proof. The last section is focused on more refined approaches to the entropy of a tileset.

\section{Presentation of the tileset}
\subsection{Source of aperiodicity} \label{sec:aperiodic}

Let us start with an observation. Consider a bi-infinite sequence $x_n$ of
positive real numbers, such that either $x_{n+1}=2x_n$ or $x_{n+1}=x_n/3$ for
every $n$.

It is easy to see that every such sequences are aperiodic. Indeed, for
all $n$ and all $k>0$, we have $x_{n+k} = x_n \times 2^i/3^j$ for
some $i,j > 0 $. If we had $x_{n+1} = x_n$, then we would have $1 = 2^i/3^j$
for $i,j > 0$, which is a contradiction.

Moreover, there exist some such sequences $x_n$ which lie in the interval
$[1/3; 2]$. Starting from some $x_0$ in this interval, we can always take
$x_{n+1}=2x_n$ if $x_n < 1$, and $x_{n+1}=x_n/3$ otherwise. The same argument
works in the opposite direction.

\subsection{Aperiodic sequences and tilings} \label{sec:seqs_tilings}

A tile is an unit square with colored sides. Consider the
(geometric) plane with a unit grid; a tiling is an assignation of a tile to
each square of the grid, in such manner that matching borders have the same
color. Thus, in a tiling, we have a bi-infinite sequence of colors along any
horizontal or vertical line of the grid.

We are going to focus on the horizontal lines of our tilings. If we use
three colors (say, $0$, $1$ and $2$) for the top and bottom sides of the
tiles, we will get bi-infinite sequences over the alphabet $\{0,1,2\}$.
Such sequences might have an \textbf{average}, i.e. a limit of averages over
finite parts as the length of the parts increases. 

Our goal is to construct a set of tiles with the following two properties:
\begin{enumerate}
  \item for every tiling, if the averages of all horizontal lines exist, they
    form a sequence $x_n$ with the property defined in
    section~\ref{sec:aperiodic}.
  \item for every such sequence $x_n$, we can find a tiling where averages
    exist and are equal to $x_n$.
\end{enumerate}

This tile set will be aperiodic. If it had a periodic tiling, it would also
have a bi-periodic tiling. In a bi-periodic tiling, all horizontal lines have
an average (due to horizontal periodicity), and form a periodic sequence (due
to vertical periodicity), which is impossible. The existence of tilings is a
consequence of the second claim.

\subsection{Tilings and automata} \label{sec:til_auto}

Our tileset should guarantee that some relation holds between all two
consecutive lines of a tiling (namely, ``$x = 2y$ or $x = y/3$''). Thus, let us
consider a stripe (a horizontal line of tiles), as displayed on
figure~\ref{fig:stripe}. We call the sequence of top numbers $a_n$, bottom
numbers $b_n$, and the matching left and right numbers $q_n$. Such a stripe can
be viewed as a run of a non-deterministic automaton, where $q_n$ are the
traveresed states, and $(a_n,b_n)$ are the input.

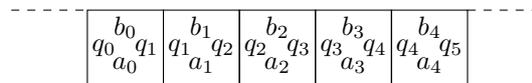
\begin{figure}[h!]
  \centering
  \begin{tikzpicture}
    \def\n{5};
    \def\nm{4};

    \draw[dashed] (-1,0) -- ++(1,0);
    \draw[dashed] (-1,1) -- ++(1,0);
    \draw[dashed] (\n,0) -- ++(1,0);
    \draw[dashed] (\n,1) -- ++(1,0);

    \foreach [count=\j] \i in {0,...,\nm}
    {
      \draw (0.5,0.5) ++ (\i,0) \tile{$q_\i$}{$b_\i$}{$a_\i$}{$q_\j$};
    }
  \end{tikzpicture}
  \caption{Example of a horizontal tiled line.}
  \label{fig:stripe}
\end{figure}

More precisely, each tile $(q',a,q,b)$ correspond to a transition $q
\xrightarrow{(a,b)} q'$, where $(a,b)$ is the input. This is illustrated by
figure~\ref{fig:translation}. Since a tileset only have a finite number of
colors and tiles, the running automaton must be a finite-state automaton
reading pairs of letters. Note that our automaton has no initial state; it runs
infinitely in both directions.

\begin{figure}[h!]
  \centering
  \begin{tikzpicture}[node distance=0.5cm and 0.5cm]
    \draw (0,0) \tile{$q_1$}{$b$}{$a$}{$q_2$};

    \node at (2,0) {$\implies$};

    \node[vertex] (q1) at (4,0) {$q$};
    \node[vertex] (q2) at (9,0) {$q'$};
    \draw[-latex,thick] (q1) to [bend left=35] node[above]{$(a,b)$} (q2);
  \end{tikzpicture}
  \caption{translation of a tile to a transition.}
  \label{fig:translation}
\end{figure}
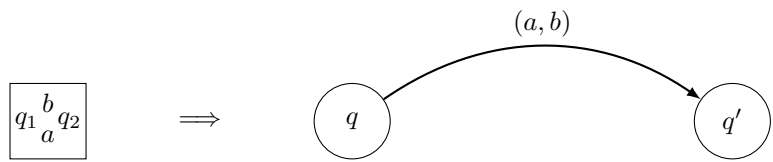

We can see that there exists a bi-infinite run of the automaton on the sequence
\begin{equation*}
  \dots (a_{-2}, b_{-2}), (a_{-1}, b_{-1}), (a_0, b_0), (a_1, b_1),
  (a_2, b_2) \dots
\end{equation*}
if and only if there exists a tile horizontal stripe that carries the sequence
$\dots a_{-2}, a_{-1}, a_0, a_1, a_2 \dots$ on the bottom and $\dots b_{-2},
b_{-1}, b_0, b_1, b_2 \dots$ on the top.

If we try to extract a set of tiles from several automata, and take the union
of the results, we will get a tileset which performs a run of one of the
automata on each line. We have to ensure that the set of states of the several
automata are disjoint, which guarantees that automata will never be mixed
within a single line.

\subsection{Construction of actual automata} \label{sec:auto}

Let us construct a finite-state automaton which reads sequences of couples
$(a_n, b_n)$ and checks if $| \sum_i b_i - 2 \, \sum_i a_i|$ is bounded;
and another one which checks if $|3 \, \sum_i b_i - \sum_i a_i|$ is bounded.
The sequences $a_n$ and $b_n$ are on an alphabet of two integers, for instance
$a_n$ is on $\{0,1\}$ and $b_n$ on $\{1,2\}$.

These automata are constructed the following way: fix a set of states $Q$,
and have all transitions $q \xrightarrow{(a,b)} q'$ to satisfy the following
relation:
\begin{align*}
  q' &= q + a - 2b \text{ (automaton for $a = 2b$)} \\
  q' &= q + 3a - b \text{ (automaton for $a = b/3$)}
\end{align*}
Thus, the automata will compute the cumulative sum of $a_n$ and the cumulative
sum of $2 b_n$ (resp. $b_n / 3$), and hold the difference into its current
state. Since the number of states is finite, the difference must be bounded.
As a consequence, if a couple of sequences $(a_n, b_n)$ is accepted by the
first (resp. second) automaton and $a_n$ have an average, then $b_n$ have an
average which is twice (resp. one third of) $a_n$'s one.

It only remains to set the alphabet for the $a_n$ and $b_n$ sequences. These
alphabets are directly connected to the allowed range for averages of $a_n$ and
$b_n$. For instance, if $a_n$ is on alphabet $\{1,2\}$, its average can be any
real number between $1$ and $2$. Likewise, we have to set an alphabet for
$b_n$. As an additional restriction, we can make automata in such manner that
they reject some finite patterns, like $000$. Sequences on alphabet $\{0,1\}$
without any pattern ``$000$'' cannot have an average lesser than $0.25$. Using
this fact, we can restrict allowed ranges for averages in a more precise way.

If one tries to build such automaton for range $[1/3; 2]$ with alphabet
$\{0,1,2\}$, he gets the automata displayed on figure~\ref{fig:automata}.

\begin{figure}[h!]
  \centering

  \begin{subfigure}{0.49\textwidth}
    \centering
    \begin{tikzpicture}
      \node[vertex] (0/3) at (0,0) {$0/3$};
      \node[vertex] (1/3) at (2,0) {$1/3$};
      \node[vertex] (2/3) at (4,0) {$2/3$};

      \draw[-latex,thick] (0/3) to[bend left=25] node[above] {$(1,0)$} (1/3);
      \draw[-latex,thick] (1/3) to[bend left=25] node[above] {$(1,0)$} (2/3);
      \draw[-latex,thick] (2/3.north) to[bend right=65] node[above] {$(1,1)$}
                          (0/3.north);

      \draw[-latex,thick] (1/3) to[bend left=25] node[below] {$(2,1)$} (0/3);
      \draw[-latex,thick] (2/3) to[bend left=25] node[below] {$(2,1)$} (1/3);
      \draw[-latex,thick] (0/3.south) to[bend right=65] node[above] {$(2,0)$}
                          (2/3.south);
    \end{tikzpicture}
  \end{subfigure}
  \hfill
  \begin{subfigure}{0.49\textwidth}
    \centering
    \begin{tikzpicture}
      \node[vertex] (0) at (0,0) {$0$};
      \node[vertex] (1) at (2,0) {$1$};

      \draw[-latex,thick] (0) to[loop above] node[above]{$(0,0)$} (0);
      \draw[-latex,thick] (0) to[loop below] node[below]{$(1,2)$} (0);
      \draw[-latex,thick] (1) to[loop above] node[above]{$(0,0)$} (1);
      \draw[-latex,thick] (1) to[loop below] node[below]{$(1,2)$} (1);
      \draw[-latex,thick] (0) to[bend left=25] node[above]{$(1,1)$} (1);
      \draw[-latex,thick] (1) to[bend left=25] node[below]{$(0,1)$} (0);
    \end{tikzpicture}
  \end{subfigure}
  \caption{automata $M_{1/3}$ and $M_2$.}
  \label{fig:automata}
\end{figure}
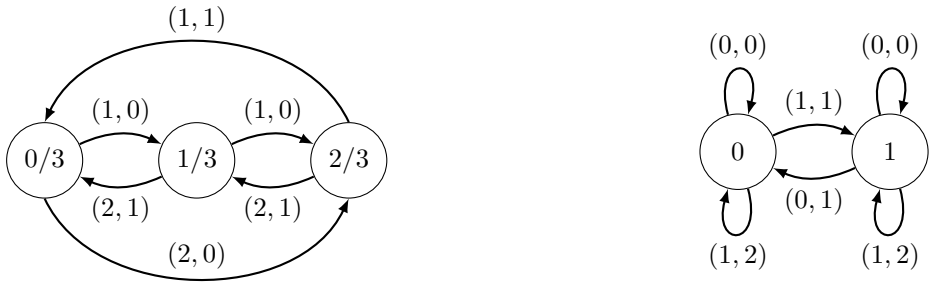

It is easy to check that these automata have $12$ transitions, yielding an
aperiodic set of $12$ tiles.

\subsection{Existence of a Tiling}

Remark that if we want to tile the whole plane, our automata cannot be fed with
any sequences, even if those sequences have averages. Indeed, it is easy to see
that when the $M_2$ automata is fed with the sequence ``$a_n=0011$'', the
sequence $b_n$ must contain both a $0$ and a $2$ to be accepted. However,
automata $M_{1/3}$ only accepts sequences over $\{1,2\}$ for $a_n$, and $M_2$
only accepts $\{0,1\}$. Thus, if a stripe has $0011$ on its bottom line, then
it has both a $0$ and a $2$ on its top line. The next stripe cannot be a run of
$M_2$ nor $M_{1/3}$, and the tiling does not go to infinity.

As a consequence, we need to show that there exists sequences of average $x$,
for each positive real $x$, which are accepted by our automata. In order to
achieve this, we will use Sturmian sequences. Define:
\begin{equation*}
  B_x(k) = \floor{x(k+1)} - \floor{xk}
\end{equation*}
$B_x$ is the \textbf{Sturmian sequence} of slope $x$, and $B_x(k)$ is its
$k^\text{th}$ letter. This sequence is bi-infinite over alphabet $\{\floor{x},
\ceil{x}\}$. Since the sum over $k$ of $B_x(k)$ is telescopic, it is easy to
calculate the average of this sequence and check it is actually $x$.

Let us think a bit about Sturmian sequences. Fix a real number $x$. Imagine you
are on an infinite, measured line, and you are making jumps of length $x$ along
the line. Whenever you make a jump, write down the number of integers you
jumped over: this is the Sturmian sequence of slope $x$. 

We can get the Sturmian sequence of $2x$ by making jumps of length $x$, and
counting the number of multiples of $0.5$ we jumped over. There are twice more
multiples of $0.5$ than integers; thus, in the long run, we actually get the
Sturmian sequence of $2x$. This idea is illustrated on
figure~\ref{fig:jumps:2}.

If we want to get the Sturmian sequence of $x/3$ from the sequence of $x$, we
have to consider multiples of $3$. They actually occur three times less often
than integers. This is illustrated on figure~\ref{fig:jumps:3}.

\begin{figure}[h!]
  \begin{subfigure}{1.0\textwidth}
    \centering
    \begin{tikzpicture}[scale=0.9]
      \def\n{15};
      \setcounter{a}{\n-1}

      % Draw the grid
      \draw[thick] (0,0) -- ++(\n,0);
      \foreach \i in {0,...,\n} {
        \draw[thick] (\i,0) -- ++(0,0.5);
      }
      \foreach \i in {0,2,...,\n} {
        \draw[thick] (\i,0) -- ++(0,0.75);
      }
      % Draw the graduations
      \foreach \i in {0,2,...,\thea} {
        \draw (\i,0) ++ (0.5,0) node[above] {1};
      }
      \foreach \i in {1,3,...,\thea} {
        \draw (\i,0) ++ (0.5,0) node[above] {0}; 
      }
      % Draw the numbers for graduations
      \foreach \i in {0,2,...,\thea} {
          \setcounter{b}{\i/2}
          \draw (\i,0) node[below] {$\theb$}
              ++(1,0)  node[below] {$\theb.5$};
      }

      % Little arrows
      \begin{scope}[yshift=0.5cm]
        \draw[-latex,thick] (0.0,0) to[bend left=45]
          node[above] {$(1)$} ++(0.8,0);
        \draw[-latex,thick] (0.8,0) to[bend left=45]
          node[above] {$(2)$} ++(0.8,0);
        \draw[-latex,thick] (1.6,0) to[bend left=45]
          node[above] {$(3)$} ++(0.8,0);
        \draw[-latex,thick] (2.4,0) to[bend left=45] ++(0.8,0);
        \draw               (3.2,0) node[right] {$\dots$};
        \draw[-latex,thick] (5.1,0) to[bend left=45]
          node[above] {$(4)$} ++(0.8,0);
        
        % Big arrows
        \draw[-latex,thick] (6.0,0) to[bend left=45]
          node[above] {$(2')$} ++(1.3,0);
        \draw[-latex,thick] (7.3,0) to[bend left=45]
          node[above] {$(3')$} ++(1.3,0);
        \draw[-latex,thick] (8.6,0) to[bend left=45] ++(1.3,0);
        \draw[-latex,thick] (9.9,0) to[bend left=45]
          node[above] {$(4')$} ++(1.3,0);
        \draw              (11.2,0) node[right] {$\dots$};
        \draw[-latex,thick] (12.9,0) to[bend left=45]
          node[above] {$(1')$} ++(1.3,0);
      \end{scope}
    \end{tikzpicture}
    \caption{Multiplication by $2$: eight types of transitions,
      but $(2)=(2')$ and $(3)=(3')$, yielding six distinct types.}
    \label{fig:jumps:2}
  \end{subfigure}

  \sk

  \begin{subfigure}{1.0\textwidth}
    \centering
    \begin{tikzpicture}[scale=0.9]
      \def\n{15}

      % Draw grid
      \draw[thick] (0,0) -- ++(\n,0);
      \foreach \i in {0,...,\n} {
        \draw[thick] (\i,0) -- ++(0,0.5);
      }
      \foreach \i in {0,3,...,\n} {
        \draw[thick] (\i,0) -- ++(0,0.75);
      }
      %Draw graduations
      \setcounter{a}{\n-1}
      \foreach \i in {0,3,...,\thea} {
          \setcounter{b}{\i/3}
          \draw (\i,0) node[below]{$\theb$}
              ++(1,0)  node[below]{$\theb.33$}
              ++(1,0)  node[below]{$\theb.66$};
      }
      \draw (\n,0) node[below]{$\dots$};

      % Draw states letters
      \setcounter{a}{\n-1}
      \foreach \i in {0,3,...,\thea} {
        \draw (\i,0) ++(0.5,0) node[above]{$\frac{2}{3}$};
      }
      \setcounter{a}{\n-1}
      \foreach \i in {1,4,...,\thea} {
        \draw (\i,0) ++(0.5,0) node[above]{$\frac{1}{3}$};
      }
      \setcounter{a}{\n-1}
      \foreach \i in {2,5,...,\thea} {
        \draw (\i,0) ++(0.5,0) node[above]{$\frac{0}{3}$};
      }

      % Draw jumps arrows
      \draw[-latex,thick] (0,0.75) to[bend left=35] 
        node[above]{$(a)$} ++(3.6,0);
      \draw[-latex,thick] (3.6,0.75) to[bend left=35]
        node[above]{$(b)$} ++(3.6,0);
      \draw[-latex,thick] (7.2,0.75) to[bend left=35]
        node[above]{$(c)$} ++(3.6,0);
      \draw[-latex,thick] (10.8,0.75) to[bend left=35]
        node[above]{$(d)$} ++(3.6,0);
    \end{tikzpicture}

    \begin{tikzpicture}[scale=0.9]
      \def\n{15}

      % Draw grid
      \draw[thick] (0,0) -- ++(\n,0);
      \foreach \i in {0,...,\n} {
        \draw[thick] (\i,0) -- ++(0,0.5);
      }
      \foreach \i in {0,3,...,\n} {
        \draw[thick] (\i,0) -- ++(0,0.75);
      }
      % Draw states letters
      \setcounter{a}{\n-1}
      \foreach \i in {0,3,...,\thea} {
        \draw (\i,0) ++(0.5,0) node[above]{$\frac{2}{3}$};
      }
      \setcounter{a}{\n-1}
      \foreach \i in {1,4,...,\thea} {
        \draw (\i,0) ++(0.5,0) node[above]{$\frac{1}{3}$};
      }
      \setcounter{a}{\n-1}
      \foreach \i in {2,5,...,\thea} {
        \draw (\i,0) ++(0.5,0) node[above]{$\frac{0}{3}$};
      }

      % Draw jumps arrows
      \draw[-latex,thick] (2.1,0.75) to[bend left=35]
        node[above]{$(e)$} ++(3.6,0);
      \draw[-latex,thick] (5.9,0.75) to[bend left=35]
        node[above]{$(f)$} ++(3.6,0);
    \end{tikzpicture}
    \caption{Division by $3$: six types of transitions.}
    \label{fig:jumps:3}
  \end{subfigure}

  \caption{multiplications of Sturmian sequences.}
  \label{fig:jumps}
\end{figure}
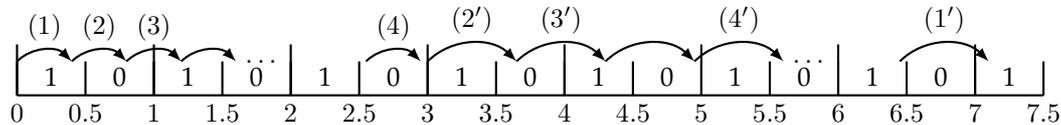
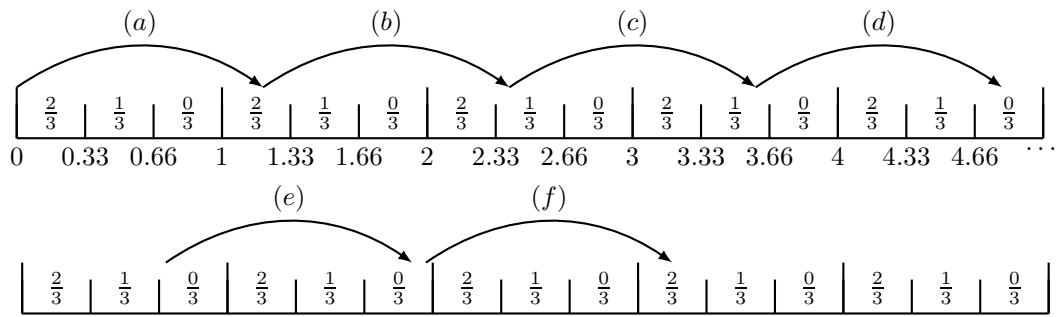

Note that, in figure~\ref{fig:jumps:2}, jumping over a non-integer
multiple of $0.5$ (the ``small obstacles'') increments the difference between
``twice number of integers jumped'' and ``number of multiples of $0.5$
jumped'' by $1$. By contrast, jumping over an integer (big obstacles)
decrements this difference by $1$. Since we want it as close to $0$ as
possible, two states are enough (before and after the small obstacle).

This works the same for figure~\ref{fig:jumps:3}. Jumping over
``small obstacles'' increments the difference between ``one third of
the obstacles jumped'' and ``number of big obstacles jumped''. Jumping
over ``big obstacles'' decrements this difference by $2$. As a conclusion,
only $3$ states are needed.

One can finally check that all possible types of jumps are displayed on
figure~\ref{fig:jumps}, and that each of them corresponds to a transition of
our automata (figure~\ref{fig:automata}). For instance, type $(1)$ corresponds
to $1 \xrightarrow{(0,0)} 1$, and $(2)$ corresponds to $1 \xrightarrow{(0,1)}
0$. More generally, in $M_2$, $a_n$ corresponds to the number of jumped
obstacles (of any size) by an arrow, and $b_n$ corresponds to the number of
jumped big obstacles. In $M_{1/3}$, it is permuted: $a_n$ is the number of
jumped big obstacles, and $M_2$ the number of jumped obstacles.

As a conclusion, $(B_x, B_{2x})$ is always accepted by $M_2$ and
$(B_{3x}, B_x)$ is always accepted by $M_{1/3}$. Thus one can take any
sequence $x_n$ from section~\ref{sec:aperiodic}, write the Sturmian
sequence of $x_n$ on line $n$ of the tiling, and get valid runs for
automata. Thus one get valid tilings.

\subsection{Aperiodicity}

This construction ensures that each tiling corresponds to a specific sequence which is identically null or aperiodic. Then we just have to avoid this null sequence. Culik presented one way to achieve that in~\cite{Culik}. The idea is to forbid three consecutive uses of the $M_2$ automaton. This can be done by adding only one tile. Consider a new color $0'$ which value is $0$ in the average, such that above a $0$ there can be either a $1$ or a $0'$, and above a $0'$ there always is a $1$.  Thus there cannot be three consecutive $not-1$ in a row, ensuring that the all zero configuration is forbidden.
All tilings with this tileset are aperiodic. This tileset is displayed on Figure~\ref{fig:tileset}

\begin{figure}
    \centering
    \begin{tikzpicture}
      \node[vertex] (0/3) at (0,0) {$\frac{0}{3}$};
      \node[vertex] (1/3) at (2,0) {$\frac{1}{3}$};
      \node[vertex] (2/3) at (4,0) {$\frac{2}{3}$};

      \foreach \i in {0,...,1}{
        \setcounter{a}{\i+1}
        \draw[-latex,thick] (\i/3) to[bend left=25] node[above] {$(2,1)$} (\thea/3);
        \draw[-latex,thick] (\thea/3) to[bend left=25] node[below] {$(1,0)$} (\i/3);
      }

      \draw[-latex,thick]  (0/3) to[bend right=65] node[below] {$(1,1)$}  (2/3);
      \draw[-latex,thick]  (2/3) to[bend right=65] node[above] {$(2,0)$}  (0/3);

      \begin{scope}[xshift= 6cm]
        \node[vertex] (0) at (0,0) {$0$};
        \node[vertex] (1) at (3,0) {$1$};

        \draw[-latex,thick] (0) to[loop above] node[above]{$(0,0')$} (0);
        \draw[-latex,thick] (1) to[loop above] node[above]{$(0,0')$} (1);

        \draw[-latex,thick] (0) to[loop below] node[below]{$(1,2)$} (0);
        \draw[-latex,thick] (1) to[loop below] node[below]{$(1,2)$} (1);

        \draw[-latex,thick] (0) to[bend left=30] node[above]{$(1,1)$} (1);
        \draw[-latex,thick] (1) to[bend left=30] node[below]{$(0',1)$} (0);
	\draw[-latex,thick] (1) to[bend left=0] node[below]{$(0,1)$} (0);
      \end{scope}

\end{tikzpicture}
\caption{Kari + Culik automata}
\end{figure}
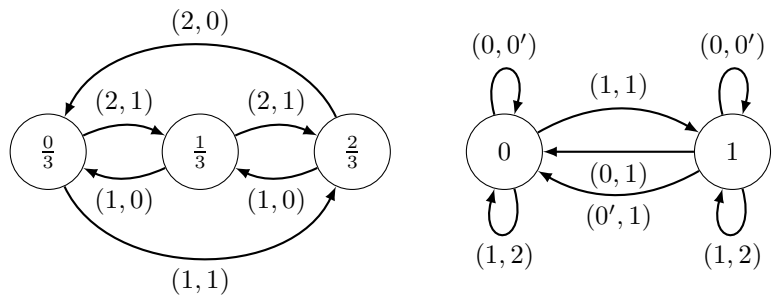

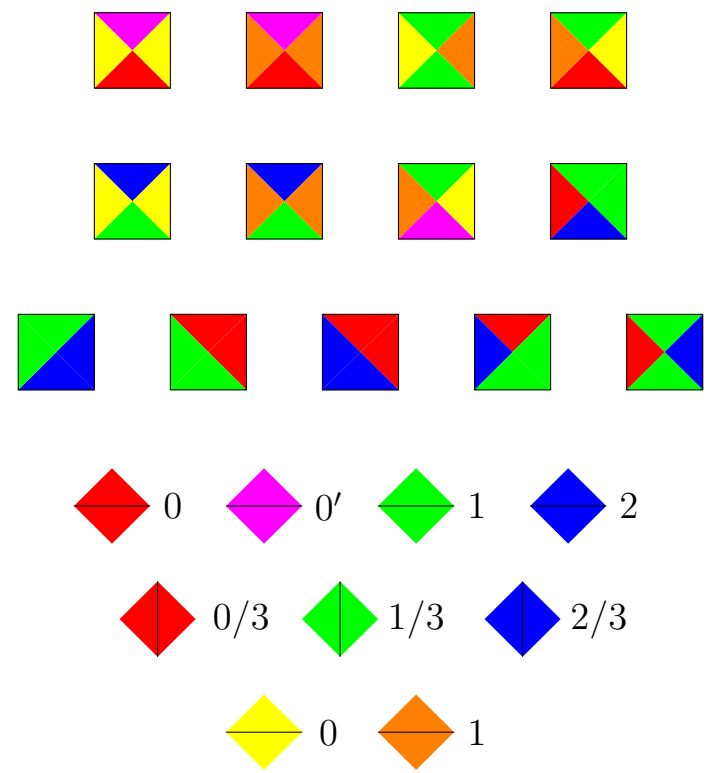
\begin{figure}[!]
\label{fig:tileset}
    \centering
    \begin{tikzpicture}
      \definecolor{letter0}{rgb}{1.0,0.0,0.0}
      \definecolor{letter1}{rgb}{0.0,1.0,0.0}
      \definecolor{letter2}{rgb}{0.0,0.0,1.0}
      \definecolor{letter0'}{rgb}{1.0,0.0,1.0}

      \definecolor{state03}{rgb}{1.0,0.0,0.0}
      \definecolor{state13}{rgb}{0.0,1.0,0.0}
      \definecolor{state23}{rgb}{0.0,0.0,1.0}

      \definecolor{state01}{rgb}{1.0,1.0,0.0}
      \definecolor{state11}{rgb}{1.0,0.5,0.0}

    %                 TOP      BOTTOM   RIGHT    LEFT
      \newwangstyle{t1}{letter0'}{letter0}{state01}{state01};
      \newwangstyle{t2}{letter0'}{letter0}{state11}{state11};
      \newwangstyle{t3}{letter1}{letter1}{state11}{state01};
      \newwangstyle{t4}{letter1}{letter0}{state01}{state11};
      \newwangstyle{t5}{letter2}{letter1}{state01}{state01};
      \newwangstyle{t6}{letter2}{letter1}{state11}{state11};
      \newwangstyle{t7}{letter1}{letter0'}{state01}{state11};

      \newwangstyle{t13}{letter1}{letter1}{state23}{state03};
      \newwangstyle{t8}{letter1}{letter2}{state13}{state03};
      \newwangstyle{t9}{letter1}{letter2}{state23}{state13};
      \newwangstyle{t10}{letter0}{letter1}{state03}{state13};
      \newwangstyle{t11}{letter0}{letter2}{state03}{state23};
      \newwangstyle{t12}{letter0}{letter1}{state13}{state23};

\draw (0,4) node[t1]{};
\draw (2,4) node[t2]{};
\draw (4,4) node[t3]{};
\draw (6,4) node[t4]{};

\draw (0,2) node[t5]{};
\draw (2,2) node[t6]{};
\draw (4,2) node[t7]{};
\draw (6,2) node[t8]{};

\draw (-1,0) node[t9]{};
\draw (1,0) node[t10]{};
\draw (3,0) node[t11]{};
\draw (5,0) node[t12]{};
\draw (7,0) node[t13]{};

\end{tikzpicture}

\vspace{1.0cm}

\begin{tikzpicture}
\definecolor{letter0}{rgb}{1.0,0.0,0.0}
      \definecolor{letter1}{rgb}{0.0,1.0,0.0}
      \definecolor{letter2}{rgb}{0.0,0.0,1.0}
      \definecolor{letter0'}{rgb}{1.0,0.0,1.0}

      \definecolor{state03}{rgb}{1.0,0.0,0.0}
      \definecolor{state13}{rgb}{0.0,1.0,0.0}
      \definecolor{state23}{rgb}{0.0,0.0,1.0}

      \definecolor{state01}{rgb}{1.0,1.0,0.0}
      \definecolor{state11}{rgb}{1.0,0.5,0.0}

\fill[letter0] (-2,0) -- (-1.5,0.5) -- (-1,0) -- (-1.5,-0.5) -- (-2,0);
\draw (-2,0) -- (-1,0);
\draw (-0.7, 0) node{\Large{$0$}};

\fill[letter0'] (0,0) -- (0.5,0.5) -- (1,0) -- (0.5,-0.5) -- (0,0);
\draw (0,0) -- (1,0);
\draw (1.35, 0) node{\Large{$0'$}};

\fill[letter1] (2,0) -- (2.5,0.5) -- (3,0) -- (2.5,-0.5) -- (2,0);
\draw (2,0) -- (3,0);
\draw (3.3, 0) node{\Large{$1$}};

\fill[letter2] (4,0) -- (4.5,0.5) -- (5,0) -- (4.5,-0.5) -- (4,0);
\draw (4,0) -- (5,0);
\draw (5.3, 0) node{\Large{$2$}};

\fill[state03] (-1.4,-1.5) -- (-0.9,-1) -- (-0.4,-1.5) -- (-0.9,-2) -- (-1.4,-1.5);
\draw (-0.9,-1) -- (-0.9,-2);
\draw (0.2, -1.5) node{\Large{$0/3$}};

\fill[state13] (1,-1.5) -- (1.5,-1) -- (2,-1.5) -- (1.5,-2) -- (1,-1.5);
\draw (1.5,-1) -- (1.5,-2);
\draw (2.5, -1.5) node{\Large{$1/3$}};

\fill[state23] (3.4,-1.5) -- (3.9,-1) -- (4.4,-1.5) -- (3.9,-2) -- (3.4,-1.5);
\draw (3.9,-1) -- (3.9,-2);
\draw (4.9, -1.5) node{\Large{$2/3$}};

\fill[state01] (0,-3) -- (0.5,-2.5) -- (1,-3) -- (0.5,-3.5) -- (0,-3);
\draw (0,-3) -- (1,-3);
\draw (1.35, -3) node{\Large{$0$}};

\fill[state11] (2,-3) -- (2.5,-2.5) -- (3,-3) -- (2.5,-3.5) -- (2,-3);
\draw (2,-3) -- (3,-3);
\draw (3.3, -3) node{\Large{$1$}};
\end{tikzpicture} 

\caption{Kari + Culik Tileset and colors meaning}
\label{fig:tileset}
\end{figure}

\section{Positive entropy} \label{sec:entropy}

\subsection{Introduction}

Let $S$ be a palette and $C_\mathcal{S}(n)$ the number of different patterns of size $n \times n$ which appear in a tiling. 
Then the \emph{entropy} of $\mathcal{S}$ is defined by $ H(\mathcal{S}) = \displaystyle \lim_{n} \frac{\log C_\mathcal{S}(n)}{n^2}$ (the limit always exists).

The question we adress is wether the Kari-Culik tileset has positive entropy. Ausual method for proving such a fact is to exhibit a substitutive pair : 
A \emph{substitutive} pair is a couple of different patterns with the same borders.

Our method is a small variant : we prove that in our tileset we have two substitutive pair and for each sufficiently large square one of the pair items appears.

Our subsitutive pairs : 

\begin{center}

\vspace{0.75 cm}

  \begin{tikzpicture}
      \definecolor{letter0}{rgb}{1.0,0.0,0.0}
      \definecolor{letter1}{rgb}{0.0,1.0,0.0}
      \definecolor{letter2}{rgb}{0.0,0.0,1.0}

      \definecolor{statep3}{rgb}{1.0,0.0,1.0}
      \definecolor{state03}{rgb}{1.0,0.0,0.0}
      \definecolor{state13}{rgb}{0.0,1.0,0.0}
      \definecolor{state23}{rgb}{0.0,0.0,1.0}

      \definecolor{state01}{rgb}{1.0,1.0,0.0}
      \definecolor{state11}{rgb}{1.0,0.5,0.0}

    %                 TOP      BOTTOM   RIGHT    LEFT
      \newwangstyle{t1}{letter2}{letter1}{state01}{state01};
      \newwangstyle{t2}{letter0}{letter2}{state03}{state23};
      \newwangstyle{t3}{letter1}{letter1}{state23}{state03};
      \newwangstyle{t4}{letter1}{letter1}{state11}{state01};

      \newwangstyle{t5}{letter0}{letter1}{state13}{state23};
      \newwangstyle{t6}{letter1}{letter2}{state23}{state13};
      \newwangstyle{t7}{letter2}{letter1}{state11}{state11};
      \newwangstyle{t8}{letter2}{letter2}{state13}{statep3};
      \newwangstyle{t9}{letter1}{letter1}{state23}{state13};
      \newwangstyle{t10}{letter1}{letter1}{state13}{state03};
      \newwangstyle{t11}{letter1}{letter1}{state03}{statep3};
      \newwangstyle{t12}{letter2}{letter1}{statep3}{state03};
      \newwangstyle{t13}{letter2}{letter1}{state03}{state13};
      \newwangstyle{t14}{letter2}{letter1}{state13}{state23};

    \begin{scope}
\draw (-1.5,0.5) node{\Large{$A_1$ =}}; 
\draw (0,0) node[t1,opacity = 0.5]{} ;
\draw (0,1) node[t2,opacity = 0.5]{} ;
\draw (1,1) node[t3,opacity = 0.5]{} ;
\draw (1,0) node[t4,opacity = 0.5]{} ;
      \draw (0,0)  \tile{$0$}{$2$}{$1$}{$0$};
      \draw (1,0) \tile{$0$}{$1$}{$1$}{$1$};
      \draw (0,1) \tile{$\frac{2}{3}$}{$0$}{$2$}{$\frac{0}{3}$};
      \draw (1,1) \tile{$\frac{0}{3}$}{$1$}{$1$}{$\frac{2}{3}$};
    \end{scope}
    \begin{scope}[xshift=4.5cm]
\draw (-1.5,0.5) node{\Large{$A'_1$ =}}; 
\draw (0,0) node[t4,opacity = 0.5]{} ;
\draw (0,1) node[t5,opacity = 0.5]{} ;
\draw (1,1) node[t6,opacity = 0.5]{} ;
\draw (1,0) node[t7,opacity = 0.5]{} ;
      \draw (0,0) \tile{$0$}{$1$}{$1$}{$1$};
      \draw (1,0) \tile{$1$}{$2$}{$1$}{$1$};
      \draw (0,1) \tile{$\frac{2}{3}$}{$0$}{$1$}{$\frac{1}{3}$};
      \draw (1,1) \tile{$\frac{1}{3}$}{$1$}{$2$}{$\frac{2}{3}$};
    \end{scope}
  \end{tikzpicture}

  \vspace{0.75cm}

  \begin{tikzpicture}
      \definecolor{letter0}{rgb}{1.0,0.0,0.0}
      \definecolor{letter1}{rgb}{0.0,1.0,0.0}
      \definecolor{letter2}{rgb}{0.0,0.0,1.0}

      \definecolor{statep3}{rgb}{1.0,0.0,1.0}
      \definecolor{state03}{rgb}{1.0,0.0,0.0}
      \definecolor{state13}{rgb}{0.0,1.0,0.0}
      \definecolor{state23}{rgb}{0.0,0.0,1.0}

      \definecolor{state01}{rgb}{1.0,1.0,0.0}
      \definecolor{state11}{rgb}{1.0,0.5,0.0}

    %                 TOP      BOTTOM   RIGHT    LEFT
       \newwangstyle{t1}{letter2}{letter1}{state01}{state01};
      \newwangstyle{t2}{letter1}{letter2}{state23}{state13};
      \newwangstyle{t3}{letter1}{letter1}{state13}{state23};
      \newwangstyle{t4}{letter1}{letter1}{state11}{state01};

      \newwangstyle{t5}{letter0}{letter1}{state03}{state13};
      \newwangstyle{t6}{letter1}{letter2}{state13}{state03};
      \newwangstyle{t7}{letter2}{letter1}{state11}{state11};
      \newwangstyle{t8}{letter0}{letter1}{state11}{state01};
      \newwangstyle{t9}{letter1}{letter1}{state23}{state03};
      \newwangstyle{t10}{letter0}{letter2}{state03}{state23};
      \newwangstyle{t11}{letter1}{letter1}{state03}{statep3};
      \newwangstyle{t12}{letter2}{letter1}{statep3}{state03};
      \newwangstyle{t13}{letter2}{letter1}{state03}{state13};
      \newwangstyle{t14}{letter2}{letter1}{state13}{state23};

    \begin{scope}
\draw (-1.5,0.5) node{\Large{$A_2$ =}}; 
\draw (0,0) node[t1,opacity = 0.5]{} ;
\draw (0,1) node[t6,opacity = 0.5]{} ;
\draw (1,1) node[t5,opacity = 0.5]{} ;
\draw (1,0) node[t4,opacity = 0.5]{} ;
      \draw (0,0) \tile{$0$}{$2$}{$1$}{$0$};
      \draw (1,0) \tile{$0$}{$1$}{$1$}{$1$};
      \draw (0,1) \tile{$\frac{0}{3}$}{$1$}{$2$}{$\frac{1}{3}$};
      \draw (1,1) \tile{$\frac{1}{3}$}{$0$}{$1$}{$\frac{0}{3}$};
    \end{scope}
    \begin{scope}[xshift=4.5cm]
\draw (-1.5,0.5) node{\Large{$A'_2$ =}}; 
\draw (0,0) node[t4,opacity = 0.5]{} ;
\draw (0,1) node[t9,opacity = 0.5]{} ;
\draw (1,1) node[t10,opacity = 0.5]{} ;
\draw (1,0) node[t7,opacity = 0.5]{} ;
      \draw (0,0) \tile{$0$}{$1$}{$1$}{$1$};
      \draw (1,0) \tile{$1$}{$2$}{$1$}{$1$};
      \draw (0,1) \tile{$\frac{0}{3}$}{$1$}{$1$}{$\frac{2}{3}$};
      \draw (1,1) \tile{$\frac{2}{3}$}{$0$}{$2$}{$\frac{0}{3}$};
    \end{scope}
  \end{tikzpicture}
\end{center}
\subsection{Coming back to the function}

Recall that our function is $f : [\frac{1}{3};2] \mapsto [\frac{1}{3};2]$ such that 
\begin{equation*}
  f(x) = \begin{cases}
    2x &\text{ if } x \in [\frac 1 3; 1] \\
    \frac 1 3 x &\text{ if } x \in [1; 2]
  \end{cases}
\end{equation*}
\begin{lemma}
The orbits of $f$ are dense.
\end{lemma}
\begin{proof}
 It is well-known that irrational rotations on the circle have dense orbits. Thus, we map the interval $[\frac{1}{3}; 2]$ on the unit circle in such manner that the function $f$ corresponds to a rotation of irrational angle.

  We consider the following mapping:
  \begin{align*}
    \phi &: [\frac{1}{3};2] \to [0;1] \\
    \phi(x) &= \frac{\log(x) + \log 3}{\log(2) + \log 3}
    \mod 1
  \end{align*}
  We view the interval $[0;1]$ as a circle by identifying point $0$ with
  point $1$.
  \begin{align}
    \phi(2x) = \frac{ \log(2) + \log(x) + \log 3 }{
    \log(2) + \log 3} \mod 1 
    &= \phi(x) + \frac{\log(2)}{\log(2) + \log 3} \mod 1
\\
    \phi(\frac{x}{3}) = \frac{\log(x)}{\log(2) + \log 3}
    \mod 1 
  &= \phi(x) + \frac{\log(2)}{\log(2) + \log 3} \mod 1
\end{align}

  Both transitions map to the same irrational rotation of angle $\frac{\log 2}{\log 2 + \log 3}$. 
\end{proof}

\begin{prop}
\label{prop:dense2}
Given any interval the maximal number of iterations of $f$ between two occurences in this interval is bounded.   
\end{prop}
\begin{proof}
Consider any interval $I = ]a;a+\alpha[$ in $[0;1] \mod 1$. As the orbits of $f$ are dense from starting point $a$, there exists $N$ such that $f^{N}(a) \mod 1$ is in $]a; a + \alpha / 2[$. Thus $f^{N}(a) = a + \beta$ with $\beta < \alpha / 2$. From any point $x$ in $I$, either $x+ \beta$ or $x - \beta$ is in $I$. Thus either $f^{N \lfloor 1 / \beta \rfloor}$ or $f^{N \lceil 1 / \beta \rceil}$ is in $I$.Hence our required bound on the number of iterations of $f$ is $N \lceil 1 / \beta \rceil$.
\end{proof}

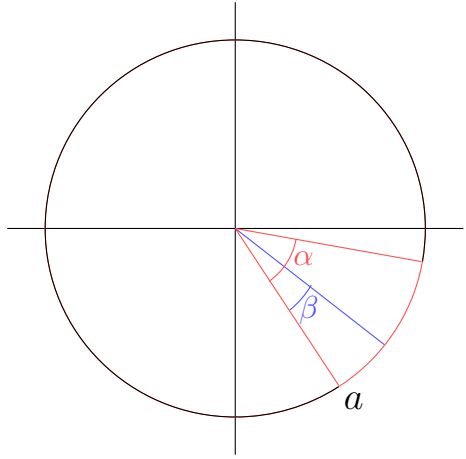
\begin{figure}
\begin{center}
\begin{tikzpicture}
\draw (-3,0) -- (3,0);
\draw (0,-3) -- (0,3);
\draw[red!70] (0,0) circle (2.5) ;
\draw (2.5,0) arc (0 : 303 : 2.5);
\draw (2.5,0) arc (0 : -10 : 2.5);
\draw[blue!70] (0,0) -- (1.96,-1.54);
\draw[blue!70] (1,-0.75) arc (-30 : -51 : 1.2) node[anchor = west]{\large{$\beta$}}  ;

\draw[red!70, thin] (0,0) -- (2.45,-0.44);
\draw[red!70, thin] (0,0) -- (1.36,-2.08);
\draw[red!70] (0.8,-0.14)  arc (-10 : -55 : 0.87)  ;
\draw[red!70] (0.9,-0.4) node{\large{$\alpha$}};
\draw (1.56,-2.28) node{\Large{$a$}};
\end{tikzpicture}
\end{center}
\label{fig:cross}
\caption{Crossing intervals}
\end{figure}

Now let us examine the colors that appear on any horizontal line of the tiling. On this line colors represent $0$ and $1$ or $1$ and $2$ ($0$ and $0'$ being interpreted as the same zero).

\begin{prop}
\label{prop:average}
Any horizontal line in any tiling has an average (in the sense of frequencies of numbers defined above). 
\end{prop}
\begin{proof}
Our proof is based on the remark that on each line we have either $0$ and $1$ or $1$ and $2$ but never $0$ and $2$.

The average of any segment has a value in $[0;2]$ which is a compact set. Consider two non overlapping growing sequences of subpatterns. Suppose their averages have different limits. Then we can take subpatterns of different averages as large as we want. With enough runs of the automata, one of this subpattern will have average less than $1$ and the other, greater than $1$, on the same line. But no automaton can read such a line, which makes a contradiction.  

\end{proof}
Remark that Kari's basic idea is that the averages of the lines obey the function $f$.

We prove below that a specific family of patterns appears dense in our tiling.
The following lemma gives us the horizontal density, and the vertical density is obtained by combination of this lemma and Proposition~\ref{prop:dense2}.

\begin{lemma}
The family of patterns $\{01^\alpha0| \alpha > 3 \}$ appear with positive density in each line that has density in $]\frac{4}{5}; \frac{9}{10}[$.
\end{lemma}
\begin{proof}
On each line with density greater than $34 / 5$ the pattern $1111$ must appear with positive density. 
On each line with density less than $9/10$, $0$ appear with a positive density, otherwise we would have a contradiction with Proposition~\ref{prop:average}.

\end{proof}

We now have a family of linear patterns that appear in a dense way in our tiling. Let us prove that each time one of this patterns appear, one element of our subsitutive pairs appear.

Let us consider the two lines above this pattern. Above the $0$'s will always be $1$'s. Above the $1$'s there will be $2$'s, expect for one $1$.
 We then distinguish three cases depending on the position of the $1$ in the block of $2$'s : the leftmost, somewhere in the middle or the rightmost. 
The line above this block of two is the result of a division by $3$. This operation is deterministic, there are a priori three posibilities for the phase of the carry on the whole line : $0$, $1/3$, $2/3$.

In the middle case (Figure~\ref{fig:cases}), the three cases for the phase of the carry are possible. In each of those cases one element of the substitutive pairs appears, either with the bottomleft or the bottomright tile being the apparition of the $1$ in the block of $2$'s. 

In the two other cases (leftmost and rightmost) we have two consecutive ones. This prevents the appearence of one of the phases (one can check that block $0011$ cannot be continued, drown in red in the pictures). 

In the leftmost case (Figure~\ref{fig:cases}), only two of the three posibilities for the phase of the carry may appear. In both cases, one element of the substitutive pairs appear above the two leftmost $1$'s of the first line. 

In the rightmost case (Figure~\ref{fig:cases}), only two of the three phases of the carry are possible. In both cases, one element of the substitutive pairs appear above the two rightmost $1$'s of the first line.

The colored vertical bar correspond to the code of bits presented in figure~\ref{fig:tileset}.

\begin{figure}[!]

\centering
    \begin{tikzpicture}[scale=0.6]
\begin{scope}
      % Colored zones and red thick lines
      \fill[color=blue!30]    (2.5,0) rectangle ++(2,2);
	\draw[color = red, very thick] (2.5,1) --++ (0,1);
      % Grid
      \draw[dotted] (-1,0.5) -- ++(8,0);
	\draw[dotted] (-1,1.5) -- ++(8,0);
      \foreach \i in {0,...,5} {
        \draw[dotted] (\i,-0.5) ++ (0.5,0) -- ++(0,3);
      }

      % Digits
      \draw (0,0) node {$0'$} 
          ++(1,0) node {$\dots$}
          ++(1,0) node {$1$}
          ++(1,0) node {$1$}
          ++(1,0) node {$1$}
          ++(1,0) node {$\dots$}
          ++(1,0) node {$0'$};
      \draw (0,1) node {$1$}
          ++(1,0) node {$\dots$}
          ++(1,0) node {$2$}
          ++(1,0) node {$1$}
          ++(1,0) node {$2$}
          ++(1,0) node {$\dots$}
          ++(1,0) node {$1$};
\draw (0,2) node {$\dots$} 
          ++(1,0) node {$\dots$}
          ++(1,0) node {$0$}
          ++(1,0) node {$1$}
          ++(1,0) node {$0$}
          ++(1,0) node {$\dots$}
          ++(1,0) node {$\dots$};
\end{scope}
\begin{scope}[xshift=9cm]
      % Colored zones and red thick lines
      \fill[color=blue!30]    (1.5,0) rectangle ++(2,2);
      \draw[color = green, very thick] (2.5,1) --++ (0,1);
	% Grid
      \draw[dotted] (-1,0.5) -- ++(8,0);
	\draw[dotted] (-1,1.5) -- ++(8,0);
      \foreach \i in {0,...,5} {
        \draw[dotted] (\i,-0.5) ++ (0.5,0) -- ++(0,3);
      }

      % Digits
            \draw (0,0) node {$0'$} 
          ++(1,0) node {$\dots$}
          ++(1,0) node {$1$}
          ++(1,0) node {$1$}
          ++(1,0) node {$1$}
          ++(1,0) node {$\dots$}
          ++(1,0) node {$0'$};
      \draw (0,1) node {$1$}
          ++(1,0) node {$\dots$}
          ++(1,0) node {$2$}
          ++(1,0) node {$1$}
          ++(1,0) node {$2$}
          ++(1,0) node {$\dots$}
          ++(1,0) node {$1$};
\draw (0,2) node {$\dots$} 
          ++(1,0) node {$\dots$}
          ++(1,0) node {$1$}
          ++(1,0) node {$0$}
          ++(1,0) node {$1$}
          ++(1,0) node {$\dots$}
          ++(1,0) node {$\dots$};

\end{scope}

\begin{scope}[xshift=18cm]
      % Colored zones and red thick lines
      \fill[color=blue!30]    (2.5,0) rectangle ++(2,2);
	\draw[color = blue, very thick] (2.5,1) --++ (0,1);
      % Grid
      \draw[dotted] (-1,0.5) -- ++(8,0);
	\draw[dotted] (-1,1.5) -- ++(8,0);
      \foreach \i in {0,...,5} {
        \draw[dotted] (\i,-0.5) ++ (0.5,0) -- ++(0,3);
      }

      % Digits
           \draw (0,0) node {$0'$} 
          ++(1,0) node {$\dots$}
          ++(1,0) node {$1$}
          ++(1,0) node {$1$}
          ++(1,0) node {$1$}
          ++(1,0) node {$\dots$}
          ++(1,0) node {$0'$};
      \draw (0,1) node {$1$}
          ++(1,0) node {$\dots$}
          ++(1,0) node {$2$}
          ++(1,0) node {$1$}
          ++(1,0) node {$2$}
          ++(1,0) node {$\dots$}
          ++(1,0) node {$1$};
\draw (0,2) node {$\dots$} 
          ++(1,0) node {$\dots$}
          ++(1,0) node {$1$}
          ++(1,0) node {$0$}
          ++(1,0) node {$1$}
          ++(1,0) node {$\dots$}
          ++(1,0) node {$\dots$};
\end{scope}
 \draw (12, -1.5) node{Middle case};
   \end{tikzpicture}

\vspace{0.5cm}

\centering
    \begin{tikzpicture}[scale=0.6]
\begin{scope}
      % Colored zones and red thick lines
      \fill[color=blue!30]    (0.5,0) rectangle ++(2,2);
	\draw[color = red, very thick] (-0.5,1) --++ (0,1);
      % Grid
      \draw[dotted] (-1,0.5) -- ++(8,0);
	\draw[dotted] (-1,1.5) -- ++(8,0);
      \foreach \i in {0,...,5} {
        \draw[dotted] (\i,-0.5) ++ (0.5,0) -- ++(0,3);
      }

      % Digits
      \draw (0,0) node {$0'$} 
          ++(1,0) node {$1$}
          ++(1,0) node {$1$}
          ++(1,0) node {$1$}
          ++(1,0) node {$1$}
          ++(1,0) node {$\dots$}
          ++(1,0) node {$0'$};
      \draw (0,1) node {$1$}
          ++(1,0) node {$1$}
          ++(1,0) node {$2$}
          ++(1,0) node {$2$}
          ++(1,0) node {$2$}
          ++(1,0) node {$\dots$}
          ++(1,0) node {$1$};
\draw (0,2) node {$1$} 
          ++(1,0) node {$0$}
          ++(1,0) node {$1$}
          ++(1,0) node {$0$}
          ++(1,0) node {$1$}
          ++(1,0) node {$\dots$}
          ++(1,0) node {$\dots$};
\end{scope}
\begin{scope}[xshift=9cm]
      % Colored zones and red thick lines
      \fill[color=blue!30]    (0.5,0) rectangle ++(2,2);
	\draw[color = green, very thick] (-0.5,1) --++ (0,1);
      % Grid
      \draw[dotted] (-1,0.5) -- ++(8,0);
	\draw[dotted] (-1,1.5) -- ++(8,0);
      \foreach \i in {0,...,5} {
        \draw[dotted] (\i,-0.5) ++ (0.5,0) -- ++(0,3);
      }

      % Digits
      \draw (0,0) node {$0'$} 
          ++(1,0) node {$1$}
          ++(1,0) node {$1$}
          ++(1,0) node {$1$}
          ++(1,0) node {$1$}
          ++(1,0) node {$\dots$}
          ++(1,0) node {$0'$};
      \draw (0,1) node {$1$}
          ++(1,0) node {$1$}
          ++(1,0) node {$2$}
          ++(1,0) node {$2$}
          ++(1,0) node {$2$}
          ++(1,0) node {$\dots$}
          ++(1,0) node {$1$};
\draw (0,2) node {$0$} 
          ++(1,0) node {$1$}
          ++(1,0) node {$0$}
          ++(1,0) node {$1$}
          ++(1,0) node {$1$}
          ++(1,0) node {$\dots$}
          ++(1,0) node {$\dots$};
\end{scope}
\begin{scope}[xshift=18cm]
      % Colored zones and red thick lines
      \fill[color=red!30]    (-0.5,1.5) rectangle ++(4,1);
	\draw[color = blue, very thick] (-0.5,1) --++ (0,1);
      % Grid
      \draw[dotted] (-1,0.5) -- ++(8,0);
	\draw[dotted] (-1,1.5) -- ++(8,0);
      \foreach \i in {0,...,5} {
        \draw[dotted] (\i,-0.5) ++ (0.5,0) -- ++(0,3);
      }

      % Digits
      \draw (0,0) node {$0'$} 
          ++(1,0) node {$1$}
          ++(1,0) node {$1$}
          ++(1,0) node {$1$}
          ++(1,0) node {$1$}
          ++(1,0) node {$\dots$}
          ++(1,0) node {$0'$};
      \draw (0,1) node {$1$}
          ++(1,0) node {$1$}
          ++(1,0) node {$2$}
          ++(1,0) node {$2$}
          ++(1,0) node {$2$}
          ++(1,0) node {$\dots$}
          ++(1,0) node {$1$};
\draw (0,2) node {$0$} 
          ++(1,0) node {$0$}
          ++(1,0) node {$1$}
          ++(1,0) node {$1$}
          ++(1,0) node {$0$}
          ++(1,0) node {$\dots$}
          ++(1,0) node {$\dots$};

\end{scope}
\draw (12, -1.5) node{Leftmost case};
    \end{tikzpicture}

\vspace{0.5cm}

\centering
    \begin{tikzpicture}[scale=0.6]
\begin{scope}
      % Colored zones and red thick lines
      \fill[color=blue!30]    (3.5,0) rectangle ++(2,2);
	\draw[color = red, very thick] (3.5,1) --++ (0,1);
      % Grid
      \draw[dotted] (-1,0.5) -- ++(8,0);
	\draw[dotted] (-1,1.5) -- ++(8,0);
      \foreach \i in {0,...,5} {
        \draw[dotted] (\i,-0.5) ++ (0.5,0) -- ++(0,3);
      }

      % Digits
      \draw (0,0) node {$0'$} 
          ++(1,0) node {$\dots$}
          ++(1,0) node {$1$}
          ++(1,0) node {$1$}
          ++(1,0) node {$1$}
          ++(1,0) node {$1$}
          ++(1,0) node {$0'$};
      \draw (0,1) node {$1$}
          ++(1,0) node {$\dots$}
          ++(1,0) node {$2$}
          ++(1,0) node {$2$}
          ++(1,0) node {$2$}
          ++(1,0) node {$1$}
          ++(1,0) node {$1$};
\draw (0,2) node {$\dots$} 
          ++(1,0) node {$\dots$}
          ++(1,0) node {$\dots$}
          ++(1,0) node {$0$}
          ++(1,0) node {$1$}
          ++(1,0) node {$0$}
          ++(1,0) node {$1$};
\end{scope}
\begin{scope}[xshift=9cm]
      % Colored zones and red thick lines
      \fill[color=red!30]    (2.5,1.5) rectangle ++(4,1);
	\draw[color = green, very thick] (3.5,1) --++ (0,1);
      % Grid
      \draw[dotted] (-1,0.5) -- ++(8,0);
	\draw[dotted] (-1,1.5) -- ++(8,0);
      \foreach \i in {0,...,5} {
        \draw[dotted] (\i,-0.5) ++ (0.5,0) -- ++(0,3);
      }

      % Digits
           \draw (0,0) node {$0'$} 
          ++(1,0) node {$\dots$}
          ++(1,0) node {$1$}
          ++(1,0) node {$1$}
          ++(1,0) node {$1$}
          ++(1,0) node {$1$}
          ++(1,0) node {$0'$};
      \draw (0,1) node {$1$}
          ++(1,0) node {$\dots$}
          ++(1,0) node {$2$}
          ++(1,0) node {$2$}
          ++(1,0) node {$2$}
          ++(1,0) node {$1$}
          ++(1,0) node {$1$};
\draw (0,2) node {$\dots$} 
          ++(1,0) node {$\dots$}
          ++(1,0) node {$\dots$}
          ++(1,0) node {$1$}
          ++(1,0) node {$1$}
          ++(1,0) node {$0$}
          ++(1,0) node {$0$};
\end{scope}
\begin{scope}[xshift=18cm]
      % Colored zones and red thick lines
      \fill[color=blue!30]    (3.5,0) rectangle ++(2,2);
	\draw[color = blue, very thick] (3.5,1) --++ (0,1);
      % Grid
      \draw[dotted] (-1,0.5) -- ++(8,0);
	\draw[dotted] (-1,1.5) -- ++(8,0);
      \foreach \i in {0,...,5} {
        \draw[dotted] (\i,-0.5) ++ (0.5,0) -- ++(0,3);
      }

      % Digits
           \draw (0,0) node {$0'$} 
          ++(1,0) node {$\dots$}
          ++(1,0) node {$1$}
          ++(1,0) node {$1$}
          ++(1,0) node {$1$}
          ++(1,0) node {$1$}
          ++(1,0) node {$0'$};
      \draw (0,1) node {$1$}
          ++(1,0) node {$\dots$}
          ++(1,0) node {$2$}
          ++(1,0) node {$2$}
          ++(1,0) node {$2$}
          ++(1,0) node {$1$}
          ++(1,0) node {$1$};
\draw (0,2) node {$\dots$} 
          ++(1,0) node {$\dots$}
          ++(1,0) node {$\dots$}
          ++(1,0) node {$1$}
          ++(1,0) node {$0$}
          ++(1,0) node {$1$}
          ++(1,0) node {$0$};

\end{scope}
\draw (12, -1.5) node{Rightmost case};

    \end{tikzpicture}

\caption{Our case analysis}
\label{fig:cases}
\end{figure}
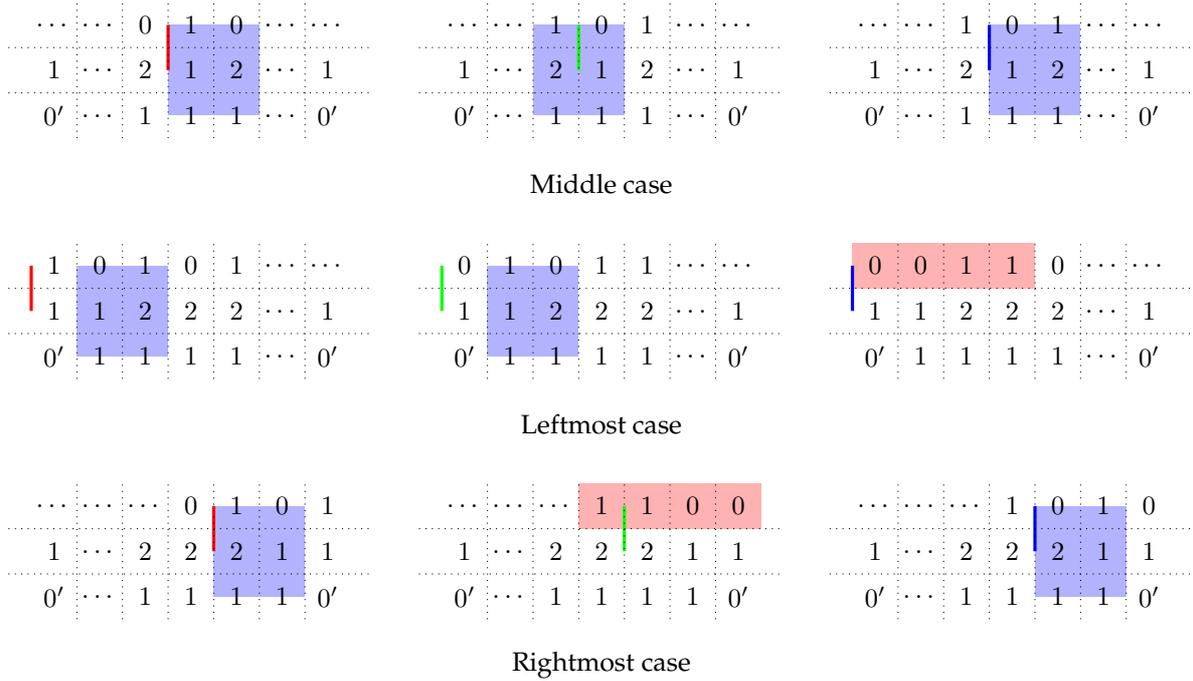

\begin{theorem}
The Kari-Culik tileset have positive entropy.
\end{theorem}
\begin{proof}
The theorem is a consequence of all the other resuts of this section : we have presented two substitutive pairs that together appear in a dense way in any tiling of the plane.
\end{proof}

\subsection{Open problems}

We proved that the two pairs are together dense in any tiling. Is one of those pairs dense alone in a given tiling? 

Consider the extended tileset where we forbid one pattern in each of the presented pairs. The obtained tileset is still a palette. Has this tileset a positive entropy? If the answer is positive, is it possible to exclude a finite number of patterns so that all the resulting tilings have zero entropy?

Using substitutive pairs we proved that there are tilings which horizontal lines do not represent mechanical words.
Is it possible to better characterize the Kari-words : the language of the lines that can appear in a tiling ?

\section{Choices, "cylindricity" and entropy}
\label{sec:choices}

Positive entropy is a very rough method for understanding the quantity of choice that you meet when you effectively construct a tiling of the plane. Imagine that you walk over the plane in spiral trajectory, placing one matching tile after another. Lets mark in red the cells were you have a "real choice" \emph{i.e.} were you have at least two possibilities that can continue to an infinite tiling of the plane. If the set of red cells is dense then the entropy of the tileset is positive. If you use this approach on selfsimilar tilesets as usually constructed, then the red cells are exponentially rare : as soon as you fix a tile on this cell you impose the next level structure and the size of the such determined areas grows exponentially. 
From the original construction of Kari~\cite{Kari} it is clear that there are horizontal lines where the density of red points is constant (because of underliyng mechanical words). Nevertheless, even if this makes a difference between self similar tilesets and Kari's, this does not prove positive entropy, furthermore it was conjectured that this freedom in representation of mechanical words of same density is strongly coupled.

Another refined version of the entropy appraoch was presented by Thierry Monteil in ~\cite{thierry} using the notion of cylindricity.
We explain below how this is related with our work. 

Consider a vertical cylinder of size $n$. If you can tile this cylinder with a tileset then two of the horizontal rings are identical, thus one can tile a torus which correspond to a periodic tiling of the plane. If the tileset is aperiodic, for each $n$ there exists a maximal vertical size for a portion of the cylinder to be tilable. The smallest growing function greater than this vertical size is called the \emph{cylindricity} function of the tileset.

Consider any self similar tileset (for instance use the generic approach of Nicolas Ollinger in~\cite{ollinger} for generating a Wang tileset from a  substitution). If one can tile a cylinder of given size, then we can rewrite all the tiles into blocks and thus obtain a larger cylinder with about the same proportions: depending on the proportion between horizontal and vertical factors. The cylindricity function is greater than $x^{\alpha}$ with $\alpha$ positive.  

Remark that the cylindricity is always greater than a logarithm : consider a tiling of the plane and a vertical segment of size $n$ in an horizontal stripe. The minimal distance for seeing twice the same vertical segment is bounded by an exponential in $n$ (because the tileset is finite). 

It would be interesting to study this function for more sophisticated tilesets, for instance the most complicated one ~\cite{Complextilings} or the robust to errors version in ~\cite{DBLP:journals/jcss/DurandRS12}. 
In Kari's tiling if you have a periodic configuration then its image after a few steps will have a period three times larger because when we divide by three, consecutive periods assume different carry phases. Note that the $\times2$ operation cannot diminish the period. Thus if we have a cylinder of length $n$ and height $h$, then the period of the first line is at most $(n/3)^{\alpha h}$ where $\alpha$ is a constant (some easy technical adjustments are needed to transform this argument into a complete proof).

From this result it was conjectured in~\cite{thierry} that the logarithm of number of patterns of size $n\times n$ was of order $n$. This would have produced an entropy zero tiling with strictly more choices than for self similar case. But it is not the case. We proved that the bound given by the cylindricity appraoch is not tight.

\section{Acknowledgments}
The authors thank Alexander Shen for his help in stating in a clear way above results.
\nocite{*}
\bibliographystyle{plain}
\bibliography{biblikari}

\end{document}